\documentclass[aps,amssymb,amsmath,amsfonts,twocolumn,pra]{revtex4}
\usepackage[latin1]{inputenc}
\usepackage{amsthm}
\usepackage{amsmath}
\usepackage{graphicx}
\usepackage{subfigure}
\usepackage{amsfonts}
\usepackage{hyperref}
\usepackage{xcolor}
\usepackage{verbatim}
\usepackage{eps-to-pdf}
\usepackage{hyperref}
\usepackage{soul}

\newtheorem{prop}{PROPOSITION}

\begin{document}
\title{Quantum measurements with prescribed symmetry}
\author{Wojciech Bruzda}
\affiliation{Institute of Physics, Jagiellonian University, Krak\'ow, Poland}
\author{Dardo Goyeneche}
\affiliation{Institute of Physics, Jagiellonian University, Krak\'ow, Poland}
\affiliation{Faculty of Applied Physics and Mathematics, Technical University of Gda\'{n}sk, 80-233 Gda\'{n}sk, Poland}
\author{Karol {\.Z}yczkowski}
\affiliation{Institute of Physics, Jagiellonian University, Krak\'ow, Poland}
\affiliation{Center for Theoretical Physics, Polish Academy of Sciences, Warsaw, Poland}

\date{July 4, 2017}

\begin{abstract}
We introduce a method to determine whether a given generalised quantum measurement is isolated or it belongs to a family of measurements having the same prescribed symmetry. The technique proposed 
reduces to solving a linear system of equations in some relevant cases. As consequence, we provide a simple derivation of the maximal family of Symmetric Informationally Complete measurement (SIC)-POVM in dimension 3. Furthermore, we show that the following remarkable geometrical structures are isolated, so that free parameters cannot be introduced: \emph{(a)} maximal sets of mutually unbiased bases in prime power dimensions from 4 to 16, \emph{(b)} SIC-POVM in dimensions from 4 to 16 and \emph{(c)} contextuality Kochen-Specker sets in dimension 3, 4 and 6, composed of 13, 18 and 21 vectors, respectively.
\end{abstract}
\maketitle

\section{Introduction}
Positive Operator Valued Measure (POVM) is the most general kind of measurement in quantum mechanics, which generalizes projective measurements. Some POVM having a prescribed symmetry play a crucial role in quantum mechanics: Symmetric Informationally Complete SIC-POVM \cite{RBSC04,FHS17} and Mutually Unbiased Bases (MUB) \cite{I81,WF89}. These geometrical structures have important applications in quantum theory: SIC-POVM and MUB allow us to unambiguously reconstruct any density matrix of size $d$ \cite{RBSC04,I81} and define entropic uncertainty relations \cite{WYM09,R13}. Even more, MUB are important to detect entanglement \cite{SHBAH12}, bound entanglement \cite{HL14}, and to lock classical information in quantum states \cite{BW07}.

Finitely many SIC-POVM are known in dimension $d\leq64$ \cite{SG10}, including a 1-parametric family in dimension three \cite{RBSC04}. The existence of SIC-POVM with free parameters in dimension $d>3$ is still an open problem. Moreover, maximal sets of MUB are known to be isolated in dimensions two to five \cite{BWB10} and finitely many non-equivalent maximal sets of MUB are known for dimensions three to five \cite{SK14} and $N$ qudit systems \cite{K12}. On the other hand, families of symmetric POVM are useful for practical applications, as the parameters can be optimized for different convenient purposes. For example, from the one-parametric family of SIC-POVM existing in dimension three \cite{RBSC04} only a single member maximizes the informational power, that is, the classical capacity of a quantum-classical channel generated by the SIC-POVM \cite{S14}. Furthermore, inequivalent sets of MUB provide different estimation of errors in quantum tomography \cite{RHKLS13}. 

Highly symmetric quantum measurements, like SIC-POVM or MUBs play an important role in quantum information and foundations of quantum theory. On one hand, it is interesting itself to design a mathematical tool that allows one to construct a family of POVM having a prescribed symmetry from a given particular solution. On the other hand, construction of such families of solutions provides flexibility when designing experimental implementations of these measurement sets. For instance, a detailed description of a complete list of  solutions of a set of $k$ MUB in dimension $d$ may be helpful in tackling the problem whether an extended set of $k+1$ MUB exists. Furthermore, it is also interesting to highlight those quantum measurements having a prescribed symmetry that do not belong to a family, which makes them special. A possible application of such isolated cases is the existence of a unique solution optimizing a given function. For instance, isolated solutions might define sets of measurement having a unique maximal violation of a Bell inequality, which is a fundamental ingredient for \emph{self-testing} \cite{MY04}.

In this work, we present a method to introduce free parameters in generalized measurements having a predefined geometrical structure. The method proposed divides the entire non-linear problem, called $\mathcal{P}^{(1)}_{NL}$, into a linear problem $\mathcal{P}^{(2)}_L$ and a secondary non-linear problem $\mathcal{P}^{(3)}_{NL}$, which is simpler than $\mathcal{P}^{(1)}_{NL}$. Remarkably, in some cases the linear problem $\mathcal{P}^{(2)}_L$ provides a definite answer to the full problem $\mathcal{P}_{NL}$.

The paper is organized as follows: in Section \ref{S2} we establish a connection between any given POVM and certain hermitian unitary matrices having constant diagonal.

In Section \ref{S3} we apply the notion of a defect of a unitary matrix to identify isolated cases of generalized 
quantum measurements having a prescribed symmetry, for which no free parameters can be introduced.
Furthermore, in other cases we present a constructive method to extend known solutions to an entire family by introducing free parameters.

In Section \ref{S4} we show that known maximal sets of MUB in dimensions 4, 8, 9 and 16 and known SIC-POVM in dimensions 4 to 16 are isolated. We also study the robustness of our results for a given accuracy in specifying the POVM, which allows us to derive conclusive results from approximate solutions. In Section \ref{S5} we find an upper bound for the maximal number of free parameters that can be introduced in sets of $2\leq m\leq d+1$ MUB in dimension $d$ and in some classes of equiangular tight frames. Moreover, we show how the method works to give the known one-parameter families of SIC-POVM in dimension three. In Section \ref{S6} we prove that some existing Kochen-Specker sets from quantum contextuality are isolated. In Section \ref{S7} we summarize our results and pose open questions.

\section{Quantum measurements and tight frames}\label{S2}
A POVM $\{\Gamma_j\}$ is a set of $N$ positive semidefinite subnormalized operators defined in dimension $d$ such that $\sum_{j=0}^{N-1}\Gamma_j=\mathbb{I}_d$, where $\mathbb{I}_d$ is the identity matrix of size $d$. Along the work, we will restrict our attention to rank-one POVM and consider rank-one projectors $\Pi_j$, being proportional to the elements of POVM. That is, $\Pi_j=c\Gamma_j$, where $c=N/d$. For simplicity, we will refer to the set of projectors $\{\Pi_j\}$ as a POVM, understanding that they are formally proportional to the elements of a POVM. The rank-one projectors $\Pi_j$ satisfy the geometrical relation
\begin{equation}\label{TrS}
\mathrm{Tr}(\Pi_i\Pi_j)=S_{ij},
\end{equation}
where $S$ is a real symmetric matrix of size $N$. It is interesting to ask about the most general projectors having the prescribed symmetry (\ref{TrS}) given by a real symmetric matrix $S$. For example, the case $S=\mathbb{I}_{N}+\frac{N-d}{d(N-1)}(\mathbb{J}_{N}-\mathbb{I}_{N})$ corresponds to equiangular tight frames composed by $N$ vectors in dimension $d$. Here, $\mathbb{J}_{N}$ denotes the matrix of size $N$ having all entries equal to unity. We recall that a set of $N$ vectors $\{|\phi_i\rangle\}$ defined in dimension $d$ forms an \emph{equiangular tight frame} (ETF) if $|\langle\phi_i|\phi_j\rangle|^2=d(N-1)/(N-d)$, for every $i\neq j=0,\dots,d^2-1$. 

A remarkably important subclass of ETF is given by the so-called \emph{Symmetric Informationally Complete} (SIC)-POVM \cite{RBSC04}, corresponding to the case $N=d^2$. Also, two orthonormal bases $|\phi_i\rangle$ and $|\psi_j\rangle$ in dimension $d$ define a pair of  \emph{mutually unbiased bases} (MUB) if $|\langle\phi_i|\psi_j\rangle|^2=1/d$, for every $i,j=0,\dots,d-1$. A set of $m$ orthonormal bases is mutually unbiased if every pair of the set is mutually unbiased. Also, a set of $m$ MUB in dimension $d$ has associated the symmetric matrix $S=\mathbb{I}_{dm}+\frac{1}{d}(\mathbb{J}_{m}-\mathbb{I}_{m})\otimes\mathbb{J}_{d}$. For a recent review on discrete structures in Hilbert spaces, including MUB and SIC-POVM, see Ref. \cite{BZ17}.

Let us recall a close connection existing between POVM and \emph{tight frames}. A set of rank-one projectors $\{\Pi_j\}$ defines a tight frame if there exists a real number $A>0$ such that $\sum_{j=0}^{N-1}\mathrm{Tr}(\Omega\Pi_j)=A\mathrm{Tr}(\Omega^2)=A$, for any rank-one projector $\Omega$ acting on dimension $d$. Therefore, POVM are tight frames for $A=c$. A crucial property for our work is the fact that the Gram matrix associated to a tight frame, or POVM, is closely related to an hermitian unitary matrix, as we will see in Proposition \ref{prop1}. We recall that the Gram matrix of a set of $N$ vectors $|\phi_j\rangle$ is given by 
\begin{equation}\label{Gram}
G_{ij}=\langle\phi_i|\phi_j\rangle,
\end{equation}
where $i,j=0,\dots N-1$. For example, the Gram matrix of an equiangular tight frame composed by $N$ vectors in dimension $d$ has the form
\begin{equation}\label{GramSIC}
G_{ETF}=\left(\begin{array}{cccc}
1&re^{i\alpha_{12}}&\dots&re^{i\alpha_{1N}}\\
re^{-i\alpha_{12}}&1&\dots&re^{i\alpha_{2N}}\\
\vdots&\vdots &\ddots&\vdots\\
re^{-i\alpha_{1N}}&re^{-i\alpha_{2N}}-1&\dots&1
\end{array}\right),
\end{equation}
where $r^2=d(N-1)/(N-d)$. Furthermore, the Gram matrix of a set of $m+1$ MUB $\{\mathbb{I}_d,H_1,H_2,\dots,H_m\}$ in dimension $d$ is given by
\begin{equation}\label{GramMUB}
G_{MUB}=\left(\begin{array}{ccccc}
\mathbb{I}_d&H_1&H_2&\cdots&H_m\\
H_1^{\dag}&\mathbb{I}_d&H_1^{\dag}H_2&\cdots&H_1^{\dag}H_m\\
H_2^{\dag}&H_2^{\dag}H_1&\mathbb{I}_d&\cdots&H_2^{\dag}H_m\\
\vdots&\vdots&\vdots&\ddots&\vdots\\
H_m^{\dag}&H_m^{\dag}H_1&H_m^{\dag}H_2&\cdots&\mathbb{I}_d\\
\end{array}\right),
\end{equation}
where $H_1,H_2,\dots,H_m$ are suitable unitary complex Hadamard matrices, so that $H_iH_i^{\dag}=\mathbb{I}$ and $|(H_i)_{jk}|^2=1/d$ for every $j,k=0,\dots,N-1$ and $i=1,\dots,m$, see \cite{TZ06}. Let us establish a connection between Gram matrices of POVM and a special kind of unitary hermitian matrices.
\begin{prop}\label{prop1}
Let $\Pi_j$ be a rank-one POVM composed by $N$ vectors in dimension $d$ and $G$ be the corresponding Gram matrix. Then, the matrix $U=\mathbb{I}_N-\frac{2d}{N}G$ is unitary.
\end{prop}
\begin{proof}
A Gram matrix $G$ represents a POVM composed by $N$ vectors in dimension $d\leq N$ if and only if $G^2=\frac{N}{d}G$ (cf. Prop. 1 in Ref.\cite{TDHS05}). From this property and taking into account that $\mathrm{Tr}(G)=N$, the spectrum of $G$ satisfies 
\begin{equation}
\lambda(G)=(\underbrace{N/d,\dots, N/d}_{d},\underbrace{0,\dots,0}_{N-d}).
\end{equation}
Therefore, $U=\mathbb{I}_N-\frac{2d}{N}G$ is a unitary matrix. 
\end{proof}
 
From Proposition (\ref{prop1}) we realize that the existence of a POVM with a symmetry prescribed by the matrix $S$ from (\ref{TrS}) is equivalent to prove the existence of a unitary hermitian matrix $U$ having positive constant diagonal $U_{ii}=1-2d/N$ and satisfying $|U_{ij}|=\frac{2d}{N}\sqrt{S_{ij}}$ for $i\neq j$. Unitary matrices $U=\mathbb{I}_N-\frac{2d}{N}G$ have been recently studied for the particular case of equiangular tight frames \cite{GT16}. In the Bloch sphere associated to a one-qubit system we have some relevant geometrical structures: orthonormal basis (line), 3 MUB (3 orthogonal lines) and SIC-POVM (tetrahedron). All these structures are unique, up to a global rotation. In higher dimensions, however, some geometrical structures allow one to introduce free parameters, that cannot be absorbed in a global rotation. In Section \ref{S3} we introduce the method, which considerably simplifies the study of this problem.

\section{Restricted defect and free parameters}\label{S3}
In this section, we derive the method to introduce the maximal possible number of free parameters into a given POVM composed by $N$ vectors in dimension $d$  associated to a given symmetric matrix $S$, see Eq.(\ref{TrS}). Using Proposition \ref{prop1}, this problem is equivalent to finding the most general real antisymmetric matrix $R$ of size $N$ such that 
\begin{equation}\label{UN}
V_{ij}(t)=U_{ij}e^{itR_{ij}},
\end{equation}
is a unitary matrix, provided that $U=\mathbb{I}_N-\frac{2d}{N}G$ is associated to a given particular POVM satisfying Eq.(\ref{TrS}). That is, $U$ is an hermitian unitary matrix having constant diagonal $U_{ii}=1-2d/N$ and $|U_{ij}|=\delta_{ij}-\frac{2d}{N}\sqrt{S_{ij}}$. Note that we have introduced a parameter $t$ in Eq.(\ref{UN}) for convenience, which can be set to $t=1$ after applying our method. The full problem is given as follows: \medskip

\noindent \textbf{Problem }$\mathbf{\mathcal{P}^{(1)}_{NL}:}$ Find the most general matrix $V_{\tau}(t)$ of size $N$ of the form (\ref{UN}), initially depending on $\tau$ parameters, such that
\begin{equation}\label{P1}
V_{\tau}(t)V_{\tau}(t)^{\dag}=\mathbb{I}_N.
\end{equation}
This problem implies to solve a system of non-linear coupled equations, which depends on $\tau=[N (N - 1)/2-z] - (N - 1)$ non-trivial variables, where $z$ is the number of zeros existing in the strictly upper triangular part of the matrix $R$. Note that $\tau$ is composed by the total number of parameters $R_{ij}$, i.e. $N (N - 1)/2-z$, minus the number of trivial variables ($N - 1$). These trivial parameters can be absorbed by applying the transformation $V\rightarrow EVE^{\dag}$, where $E=\mathrm{Diag}(1,e^{itR_{01}},\dots,e^{itR_{0(N-1)}})$. In order to simplify the resolution of problem $\mathcal{P}_{NL}$ we define the following linear problem:\medskip

\noindent\textbf{Problem }$\mathbf{\mathcal{P}^{(2)}_L:}$ Find the most general matrix $V_{\tau}(t)$ of size $N$, initially depending on $\tau$ parameters, such that
\begin{equation}\label{P2}
\lim_{t\rightarrow0}\frac{d}{dt}[V_{\tau}(t)V_{\tau}(t)^{\dag}]=0.
\end{equation}
Using (\ref{UN}), we can explicitly write Eq.(\ref{P2}) as
\begin{equation}\label{soe}
-2V_{k,k}V_{k,j}R_{j,k}+\sum_{l\neq j,k}V_{k,l}V_{l,j}(R_{k,l}-R_{j,l})=0,
\end{equation}
for $1\leqslant j<k\leqslant N$ and $1\leq l\leq N$, which is a linear problem on variables $R_{ij}$.
Note that $\mathcal{P}^{(2)}_L\subset\mathcal{P}^{(1)}_{NL}$, as Eq.(\ref{P2}) is a necessary condition to obtain Eq.(\ref{P1}). 

The linear problem  $\mathcal{P}^{(2)}_L$ allows us to simplify the full problem $\mathcal{P}_{NL}$ by determining $r$ out of $\tau$ variables $R_{ij}$, where $r$ is the number of linearly independent equations (\ref{soe}). After solving $\mathcal{P}^{(1)}_{L}$, the remaining number of free parameters $R_{ij}$ lead us to the definition of the \emph{restricted defect} $\Delta$ of the hermitian unitary matrix $U$. It reads,
\begin{equation}\label{Rd}
\Delta=\tau - r,
\end{equation}
where $\tau=(N-1)(N-2)/2-z$ and $z$ is the number of zeros existing in the strictly upper triangular part of the matrix $U$. Note that this quantity can be considered as a defect of a unitary matrix \cite{TZ08},
adopted to the case of matrices with a special structure.
The standard defect was used to define an upper bound on the number of free parameters allowed by complex Hadamard matrices \cite{TZ06} and forms, by construction, an upper bound for the restricted defect.
In both cases, the defect equal to zero implies that a given solution is isolated, 
so no free parameters can be introduced.

In general, the restricted defect represents an upper bound for the maximal number of free parameters allowed by the full problem $\mathcal{P}_{NL}$. If $\Delta=0$, then the full problem $\mathcal{P}_{NL}$ is solved by the linear problem $\mathcal{P}^{(1)}_{L}$. In this case, it is not possible to introduce free parameters into the matrix $V$. On the other hand, if $\Delta>0$, it is necessary to solve an additional non-linear problem in order to determine the continuous family of solutions.\medskip

\noindent\textbf{Problem }$\mathbf{\mathcal{P}^{(3)}_{NL}:}$ Find the most general matrix $V_{\Delta}(t)$ of size $N$, initially depending on $\Delta$ parameters, such that
\begin{equation}\label{P3}
V_{\Delta}(t)V_{\Delta}(t)^{\dag}=\mathbb{I}_N.
\end{equation}
Note that Problem $\mathcal{P}^{(2)}_{NL}$ is simpler than Problem $\mathcal{P}_{NL}$ as $\Delta<\tau$. This is so because $r>0$ in Eq.(\ref{Rd}). 
After solving Problem $\mathcal{P}^{(2)}_{NL}$ we can assume that $t=1$, without loss of generality.
 In Section \ref{S4} we will apply our results to SIC-POVM and maximal sets of MUB.\medskip

Let us first illustrate the method in action by considering two MUB for a single qubit system: $|\phi_i\rangle=|i\rangle$, $i=0,1$ and $|\psi_{\pm}\rangle=(|0\rangle\pm|1\rangle)/\sqrt{2}$. The Gram matrix (\ref{GramMUB}) associated to this set of $m=2$ MUB is given by
\begin{equation}
G_{MUB}=\frac{1}{2}\left(\begin{array}{cccc}
2&0&1&1\\
0&2&1&-1\\
1&1&2&0\\
1&-1&0&2
\end{array}\right).
\end{equation}
A family of two MUB stemming from this fixed set would have associated a Gram matrix of the form
\begin{equation}
\label{G4MUB}
G_{MUB}=\frac{1}{2}\left(\begin{array}{cccc}
2&0&e^{itR_{13}}&e^{itR_{14}}\\
0&2&e^{itR_{23}}&-e^{itR_{24}}\\
e^{-itR_{13}}&e^{-itR_{23}}&2&0\\
e^{-itR_{14}}&-e^{-itR_{24}}&0&2
\end{array}\right),
\end{equation}
and, from Prop.(\ref{prop1}), the unitary matrix $U=\mathbb{I}_N-\frac{2d}{N}G_{MUB}$. Note that full problem $\mathcal{P}^{(1)}_{NL}$ initially depends on $\tau=1$ non-trivial parameter, as $R_{13}$, $R_{14}$ and $R_{23}$ can be absorbed by considering the diagonal unitary operator $E=\mathrm{diag}[1,e^{-iR_{23}},e^{iR_{13}},e^{iR_{14}}]$ and the redefinition $V\rightarrow EVE^{\dag}$. Therefore, according to Eq.(\ref{UN}), and after considering the diagonal transformation $E$ we find that
\begin{equation}\label{U}
V_{\tau}(t)=\frac{1}{2}\left(\begin{array}{cccc}
0&0&1&1\\
0&0&1&-e^{itR_{24}}\\
1&1&0&0\\
1&-e^{-itR_{24}}&0&0
\end{array}\right).
\end{equation}
 Problem $\mathcal{P}^{(2)}_{L}$ implies the following equation
\begin{equation}
R_{24}=0,
\end{equation}
where $r=1$ and, therefore, $\Delta=0$. Thus, we cannot introduce free parameters in Eq.(\ref{U}), which implies that the considered pair of MUB is isolated. Indeed, $|\phi_i\rangle=|i\rangle$, $i=0,1$ and $|\psi_{\pm}\rangle=(|0\rangle\pm|1\rangle)/\sqrt{2}$  is the unique pair of MUB existing in dimension 2, up to a global rotation \cite{BWB10}.

\section{Isolated MUB and SIC-POVM}\label{S4}
In this section, we study the problem of introducing free parameters in MUB and SIC-POVM. Our first main result consists in proving that maximal sets of MUB existing in low prime power dimensions.
\begin{prop}\label{prop2}
Maximal sets of $d+1$ MUBs in dimensions $d=4, 8, 9$ and $16$ are isolated.
 \end{prop}
The results arises as follows. The upper triangular part of the Gram matrix associated to a set of $m$ MUB in dimension $d$ contains $z=md(d-1)/2$ zero entries. Given that the size of the Gram matrix of the set is $N_{MUB}=md$ and $m=d+1$ we have $z=d(d^2-1)/2$. Therefore, the matrix $G$ contains $\tau=(N_{MUB}-1)(N_{MUB}-2)/2-z$ parameters. The next step consists in determining how many of these parameters are remaining after solving the linear problem $\mathcal{P}^{(1)}_{L}$. To this end, we calculated the number of linearly independent equations of the linear system defined in Eq.(\ref{soe}) for the cases $d=4,8,9$ and $16$, finding $r_4=141$, $r_8=2233$ and $r_9=3556$ and $r_{16}=34545$, respectively. By using these results and Eq.(\ref{Rd}) we find that the restricted defect $\Delta$ vanished in all these cases. We based our calculation of the restricted defect $\Delta$ on the maximal sets of MUB provided in Refs. \cite{WF89,SK14}. In dimension $d=9$, we considered Ref. \cite{C02} to obtain simpler expressions of results presented in \cite{WF89}.\medskip

Let us now consider the case of SIC-POVM. It is well-known that SIC-POVM for a qubit system is essentially unique, as it represents the regular tetrahedron inscribed into the Bloch sphere, up to a global rotation. Furthermore, single parameter families of SIC-POVM for a qutrit system exists \cite{RBSC04}, which represent the most general solution \cite{FS14}. For higher dimensions, the problem of introducing free parameters in SIC-POVM is still open. As a preliminary result, exhaustive numerical simulations indicate that free parameters cannot be introduced in SIC-POVM, at least in low dimensions higher than three. For this problem, we have solved the linear problem $\mathcal{P}^{(1)}_{L}$ in dimensions $d=4,\dots,16$, obtaining the following results.
\begin{prop}\label{prop3}
SIC-POVM  in dimensions $d=4,\dots,16$ are isolated.
 \end{prop}
This result also includes the \emph{Hoggar lines} \cite{H98}, a special class of SIC-POVM defined in dimension $d=8$. We considered the total number of parameters $\tau=(N_{SIC}-1)(N_{SIC}-2)/2$, as there is no pair of orthogonal vectors in SIC-POVM ($z=0$), and $N_{SIC}=d^2$. In order to prove Proposition \ref{prop3} we solved the linear problem $\mathcal{P}^{(1)}_{L}$ for both analytical  \cite{Z99,G04,A05,H98,G06} and highly accurate numerical SIC-POVM \cite{SG10}. In all the cases we have found $\Delta=0$, which implies that free parameters cannot be introduced. Calculations of the restricted defect have been done in Matlab.

Let us now study the robustness of the restricted defect $\Delta$ under the presence of inaccuracies in describing the POVM. Given the set of vectors $\{\phi_j\}$ defined in Eq.(\ref{Gram}) we quantify the inaccuracy in approximate vectors $\{\phi^{\prime}_j\}\approx\{\phi_j\}$ by introducing the inaccuracy factor:
\begin{equation}\label{inac}
s=\frac{1}{\sqrt{d}}\max_j \|\phi^{\prime}_j-\phi_j\|.
\end{equation}
The factor $s_{\mu}$ quantifies the maximal allowed inaccuracy in entries of vectors $\phi_j$. For example, in the case of approximate solutions $\{\phi^{\prime}_j\}$ having $k$ digits of precision we have $s\approx10^{-k}$. In our study of robustness of the defect $\Delta$, we simulate the introduction of inaccuracies by considering
\begin{equation}
(\phi^{\prime}_j)_i=(\phi_j)_i+s\xi_i,
\end{equation}
where $(x)_i$ denotes the $i$-th entry of the vector $x$ and $\xi_i$ are random numbers uniformly distributed in the interval $[-1,1]$. 

Let us assume that $\mathcal{R}$ is the real matrix associated to the linear system of equations (\ref{soe}). Note that the number of linearly independent equations of such system is given by $\mathrm{rank}(\mathcal{R})=r$. When considering inaccuracies in the POVM the rank of the perturbed matrix $\mathcal{R}^{\prime}$ and $\mathcal{R}$ may differ. Therefore, we need to study how much the singular values of $\mathcal{R}^{\prime}$ are affected under the presence of inaccuracies. In particular, we are interested on the perturbation of the two smallest singular values $\sigma_0=0$ and $\sigma_1>0$, which are the responsible for the variation of the rank. In order to obtain a confidence region for the restricted defect (\ref{Rd}) we need to consider the following two bounds: (\emph{i}) upper bound for the maximal perturbation of $\sigma_0$ and (\emph{ii}) lower bound for the maximal perturbation of $\sigma_1$. In Appendix \ref{Ap1} we show that
\begin{equation}\label{bound}
|\sigma^{\prime}_i-\sigma_i|\leq f(d,N)\,s,
\end{equation}
for $i=0,1$, where $s$ is the inaccuracy quantificator defined in Eq.(\ref{inac}) and 
\begin{equation}
f(d,N)=\frac{2^6d^2}{N}\left(1-\frac{2d}{N}\right)^2\sqrt{\frac{N-d}{N(N-1)}}
\end{equation}
for $N>2d$. Let us now find the smallest possible value of $s$ such that the critical condition $\sigma^{\prime}_0=\sigma^{\prime}_1$ holds, which imposes an upper bound for the confidence region of the restricted defect $\Delta$. By considering Eq.(\ref{bound}) we find that $\Delta$ does not change its value for $0\leq s\leq s_{max}$, where $s_{max}=\sigma_1(2f(d,N))^{-1}$. Note that $\sigma_1$ depends on the exact solution, which is not known if the exact solution is not available. By using Eq.(\ref{bound}) we find that $\sigma_1\leq \sigma^{\prime}_1+f(d,N)\,s$, which implies
\begin{equation}
s_{max}\leq\frac{\sigma^{\prime}_1+f(d,N)\,s}{2f(d,N)}.
\end{equation}
Note that this inequality provides a confidence region only if $f(d,N)\,s\ll 1$.

Confidence regions for SIC-POVM and maximal sets of MUB for $2$ and $3$ qubits are depicted in Fig. \ref{Fig1} and Fig. \ref{Fig2}, respectively. For the case of $4$ qubit systems we have solutions with precision $s=10^{-32}$, whereas the upper bound for the confidence region is $s_4\approx4\times 10^{-3}$. For MUB, we calculated the restricted defect $\Delta$ by considering analytic solutions in all the cases \cite{WF89,SK14,C02}.

\begin{figure}[!ht]
\includegraphics[width=.45\textwidth]{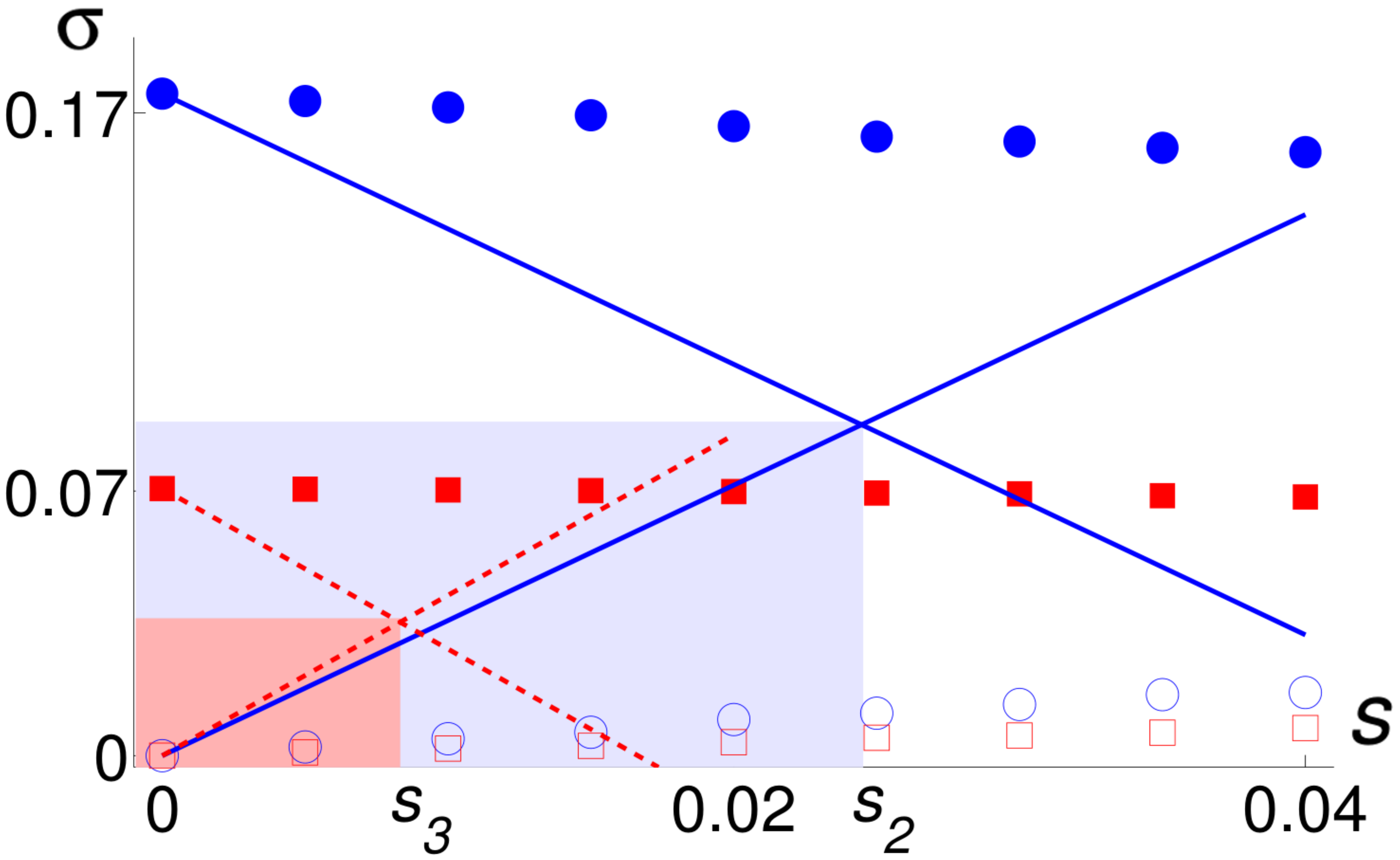}
\caption{Smallest singular values $\sigma_0$ and $\sigma_1$ of $\mathcal{R}$ as a function 
of the inaccuracy factor $s$ for SIC-POVM. Contour and filled symbols represent $\sigma_{0}$ and $\sigma_{1}$, respectively, for $2$ ($\circ$) and $3$ ({\tiny $\square$}) qubit systems.
Each case is averaged over $8$ samples randomly chosen and generated from the approximate SIC-POVM provided in Ref. \cite{SG10}, which has accuracy $s=10^{-32}$. The confidence regions (blue and red rectangles) are given by values of $s$ existing between zero and the value determined by the intersection of the lower and upper bounds. Here, $s_2$ and $s_3$ stand for $2$ and $3$-qubits systems, respectively. Outside the confidence regions it is not possible to discriminate between singular values $\sigma_0$ and $\sigma_1$.}
\label{Fig1}
\end{figure}

\begin{figure}[!ht]
\includegraphics[width=.45\textwidth]{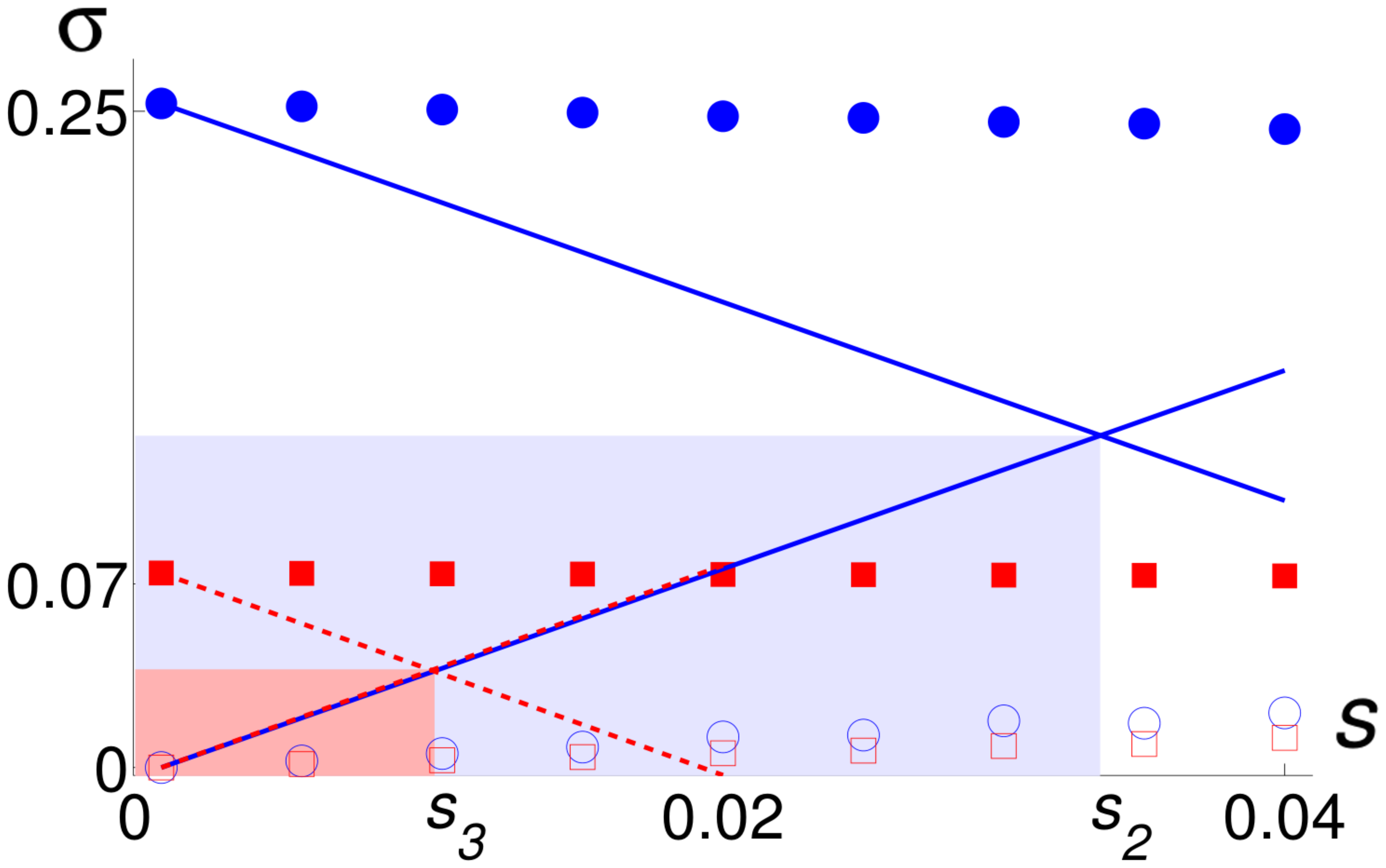}
\caption{Smallest singular values $\sigma_0$ and $\sigma_1$ of $\mathcal{R}$ as a function 
of the inaccuracy factor $s$ for maximal sets of MUB. Contour and filled symbols represent $\sigma_{0}$ and $\sigma_{1}$, respectively, for $2$ ($\circ$) and $3$ ({\tiny $\square$}) qubit systems.
Each case is averaged over $8$ samples randomly chosen and generated from analytic solutions.}
\label{Fig2}
\end{figure}

\begin{table}[!ht]
\begin{tabular}{c|cccccccc}
     $m\setminus d$  & 2 & 3 & 4 & 5 & 6 & 7 &  8 &  9 \\ 
\hline
  2  & 0 & 0 & 3 & 0 & 4 & 0 & 21 & 16 \\
  3  & 0 & 0 & 3 & 0 & $\diamond$ & 0 & 27 & 20\\
  4  & - & 0 & 0 & 0 & ? & 0 & 19 & 32\\
  5  & - & - & 0 & 0 & ? & 0 &  7 &  0\\
  6  & - & - & - & 0 & ? & 0 &  0 &  0\\
  7  & - & - & - & - & ? & 0 &  0 &  0\\
  8  & - & - & - & - & - & 0 &  0 &  0\\
  9  & - & - & - & - & - & - &  0 &  0\\
 10  & - & - & - & - & - & - &  - &  0\\
\end{tabular}
\caption{Upper bound on the maximal number of free parameters $\Delta$ allowed
by subsets of $m$ MUB in dimension $d$. The results do not depend on selecting $m$ subsets of MUB out of the full set of $d+1$ MUB. As a remarkable observation, maximal sets of
MUB are isolated. Also, subsets of $m\geq6$ in dimension 8 and $m\geq5$ MUB in
dimension 9 are isolated in all the cases.
Question marks denote our lack of knowledge about given number of MUB, while $\diamond$ indicates the case of three MUB in dimension six
which is still considered as unresolved.}
\label{T1}
\end{table}

\section{Free parameters in MUB and SIC-POVM}\label{S5}
In Section \ref{S4} we have proven that some maximal sets of MUB and SIC-POVM are isolated in low dimensions. In this section, we first calculate the restricted defect for sets of  $m=2,\dots,d+1$ MUB in dimensions $d=2,\dots,8$ (see Table \ref{T1}). Interestingly, the restricted defect for a pair of MUB $B_1$ and $B_2$ coincides with the standard defect of the complex Hadamard matrix $H=B_1^{\dag}B_2$. Indeed, the maximal number of free parameters that can be introduce in the Gram matrix
\begin{equation}\label{G2MUB}
G_{MUB}=\left(\begin{array}{cc}
\mathbb{I}_d&H\\
H^{\dag}&\mathbb{I}_d
\end{array}\right),
\end{equation}
coincides with the maximal number of parameters that can be introduced in $H$. Let us explain some details concerning Table \ref{T1}. First, note that $\Delta=0$ for every subset of $2\leq m\leq d+1$ MUB in prime dimensions $d=2,3,5,7$. This is so because every complex Hadamard matrix involved in the set is equivalent to the Fourier matrix, which is isolated in prime dimensions \cite{TZ06}. For triplets in dimension 4, we have $\Delta=3$, which coincides with the maximal number of free parameters that can be introduced \cite{BWB10}. Generic restricted defect for a pair of MUB in dimension 6 is four, coinciding with generic defect of complex Hadamard matrices of size 6 \cite{TZ08}. However, note that there is an exceptional pair of MUB for which $\Delta=0$, as an isolated complex Hadamard matrix of size 6 exists \cite{BBELTZ}.
Generic defect for triplets of MUB is not well understood ($\diamond$), see Table \ref{T1} and Ref. \cite{G13}. 

In dimensions 8 and 9 we restricted our attention to subsets of MUB arising from the maximal sets defined in Refs. \cite{WF89,SK14}. Note that subsets of $m\geq6$ MUB are isolated in dimension 8, whereas several families of $m=5$ MUB exist \cite{GG15}. Another observation is that the restricted defect $\Delta$ for maximal sets of $d+1$ MUB in dimension $d$ coincides with the defect for $d$ MUB. This is so because the ($d+1$)-th MUB is univocally determined by the first $d$ MUB. For subsets of $m<d$, the restricted defect for $m$ MUB may depend on the subset chosen. However, results presented in Table \ref{T1} are consistent for every subset of MUB. \medskip

We also studied the restricted defect for equiangular tight frames composed by $N=k^2$ vectors in dimension $d=k(k-1)/2$, typically denoted as ETF($d,N$) \cite{FM15}. These ETF have associated the following hermitian unitary matrices \cite{GT16}:
\begin{equation}\label{UFourier}
U_{i_1+ki_2+1,j_1+kj_2+1}=\omega^{i_1j_2-j_1i_2},
\end{equation}
where $i_1,i_2,j_1,j_2\in\{0,..,k-1\}$ and $\omega=e^{2\pi i/k}$. Matrix (\ref{UFourier}) is equivalent to the tensor product of Fourier matrices, $F_k\otimes F_k$, where $(F_k)_{st}=\omega^{st}$. Table \ref{T2} summarizes the restricted defect for matrix $U$ in low dimensions. Results are shown in Table \ref{T2}. 
For prime values of $k$, the formula
\begin{equation}
\Delta=\frac{1}{2}(k+1)(k-1)(k-2),
\end{equation}
matches all solutions presented in Table \ref{T2}, so we are tempted to
believe that it holds for any prime $k$.

\begin{table}[!ht]
\begin{tabular}{c|c||c|c||c|c}
$k$&$\Delta$&$k$&$\Delta$&$k$&$\Delta$\\
\hline
2&0&7&120&12&1237\\
3&4&8&273&13&924\\
4&21&9&352&14&1632\\
5&36&10&576&&\\
6&112&11&540&&
\end{tabular}
\caption{Maximal number of free parameters that can be introduced in ETF composed by $N=k^2$ vectors in dimension $d=k(k-1)/2$. The case $k=2$ corresponds to an ETF(3,4) which is isolated (regular simplex in dimension 3). Also, $k=3$ has associated a SIC-POVM in dimension 3, where  $\Delta=4$ but only 2-parametric families exist \cite{RBSC04}. For the case $k=4$ there exists several 6-parametric families of  ETF(6,16) \cite{GT16}.}
\label{T2}
\end{table}

\medskip

Let us now study the SIC-POVM problem in dimension $d=3$. By considering the fiducial state $|\phi_{00}\rangle=(1,-1,0)/\sqrt{2}$ a SIC-POVM is given by \cite{Z99}:
\begin{equation}
|\phi_{st}\rangle=X^sZ^t|\phi_{00}\rangle,
\end{equation}
where $X|j\rangle=|[j\oplus1]\rangle$, $Z|j\rangle=\omega^j|j\rangle$, $\omega=e^{2\pi i/3}$ and $\oplus$ means addition modulo 3. The $9\times 9$ Gram matrix $G_{SIC}$ of the SIC-POVM and its associated unitary matrix $U=\mathbb{I}_N-\frac{2d}{N}G_{SIC}$ depends on $\tau=36-8=28$  parameters, where the $N-1=8$ trivial parameters $R_{1,j}$ for $j=2,\dots,9$ have been set as zero. The linear system of equations (\ref{P2}) associate to problem $\mathcal{P}^{(1)}_{L}$   has $r=24$ linearly independent equations, which provides a $4$-dimensional complex set of solutions $R_{ij}$, depending on four parameters: $R_{23},R_{26},R_{48},$ and $R_{89}$. The additional restriction to have real parameters imply 
\begin{equation*}
R_{23} - 3R_{89} =R_{23} - 3 R_{26} = R_{89}-R_{26} = 0,
\end{equation*}
which is equivalent to $R_{26}=R_{89}=R_{23}/3$. After setting $t=1$, we obtain three solutions to problem $\mathcal{P}^{(1)}_{L}$: $V_{\Delta}(R_{23},R_{48})$, $V_{\Delta}(R_{26},R_{48})$ and $V_{\Delta}(R_{89},R_{48})$. Now, we are in position to solve the non-linear problem $\mathcal{P}^{(2)}_{NL}$, which is much simpler than the full non-linear problem $\mathcal{P}_{NL}$. Indeed, $\mathcal{P}^{(2)}_{NL}$ implies to solve trivial trigonometric equations, which give us the solutions $R_{26}\in\{0,\pi\}$, $R_{89}\in\{0,\pi\}$ and $R_{23}\in\{0,\pi\}$, respectively. Therefore, we generate six 1-parametric families of SIC-POVM in dimension three:
\begin{eqnarray}\label{SIC3}
&\hspace{-0.5cm}\mathcal{S}_1: V_{\Delta}(R_{23}=0,R_{48}),\hspace{0.2cm}\mathcal{S}_2: V_{\Delta}(R_{23}=\pi,R_{48}),&\nonumber\\
&\hspace{-0.5cm}\mathcal{S}_3: V_{\Delta}(R_{26}=0,R_{48}),\hspace{0.2cm}\mathcal{S}_4: V_{\Delta}(R_{26}=\pi,R_{48}),&\nonumber\\
&\hspace{-0.5cm}\mathcal{S}_5: V_{\Delta}(R_{89}=0,R_{48}),\hspace{0.2cm}\mathcal{S}_6: V_{\Delta}(R_{23}=\pi,R_{48}).&
\end{eqnarray}
Here we note that $\Delta=4$ and six 1-parametric real solutions exist. These six solutions belong to the 4-dimensional tangent plane defined by Eq.(\ref{P2}) and do not fit into a lower dimensional tangent space, which explains why $\Delta$ cannot take a lower value. Furthermore, solutions (\ref{SIC3}) are equivalent, in the sense that we can transform one into the other by applying permutation of rows or columns and multiplication of diagonal unitary operations to the Gram matrix, which is equivalent to relabel and apply global phases to vectors. Solution (\ref{SIC3}) represents the most general SIC-POVM existing in dimension three \cite{FS14}, up to equivalence. We remark that the generic hermitian defect for a SIC-POVM in dimension 3 is $\Delta=2$, with the only exception of the particular vector $|\phi_{00}\rangle$, where $\Delta=4$, however from this fact we cannot define a larger family.

\section{Isolated Kochen-Specker sets}\label{S6}
In this section, we apply our method presented in Section \ref{S3} to show that some sets of vectors used in a proof of the Kochen-Specker contextually theorem \cite{KS67}, typically called KS sets, are isolated. 

KS sets are collections of $N$ vectors in dimension $d$, which contain $m$ subsets of $d$ vectors forming orthonormal basis. Some of these orthonormal bases have common vectors, so that $N< md$. These intersections are crucial to prove that a deterministic local hidden variable theory is not possible \cite{KS67}. That is, for a system prepared in a quantum state $\rho$ and a set of KS vectors $\{\phi_0,\dots,\phi_{N-1}\}$ it is not possible to end up with $N$ deterministic probabilities $P_k=\mathrm{Tr}(\rho|\phi_k\rangle\langle\phi_k|)\in\{0,1\}$, for $k=0,\dots,N-1$. Therefore, the assumption of hidden determinism in quantum mechanics implies that predefined values of observables depend on the context in which measurements were implemented. The original proof given by Kochen and Specker involves $N=117$ vectors in dimension $d=3$ \cite{KS67}. Subsequently, examples exhibiting a lower number of vectors were found. Some remarkable examples are KS sets composed by $N=13$ vectors in dimension $d=3$  (Yo and Oh \cite{SO12}), $N=18$ vectors in dimension $d=4$ (Cabello \emph{et al.} \cite{CEG96}) and $N=21$ vectors in dimension $d=6$ (Lisonek \emph{et al.} \cite{LBPC14}). 

Let us now apply our method to prove that these three inequivalent KS sets are isolated. The first important observation is that the three KS sets form three POVM. This means that Proposition \ref{prop1} holds for these sets and, therefore, the method to introduce free parameters presented in Section \ref{S3} can be applied. In order to do so we have to calculate the restricted defect $\Delta$ defined in Eq. (\ref{Rd}), which is a function of the total number of parameters $\tau$ and the number of linearly independent equations associated to Problem $\mathbf{\mathcal{P}^{(2)}_L}$ (see Section \ref{S3}). The geometrical structure is determined by the orthogonality restrictions imposed by the KS sets. For the above mentioned three KS sets we have shown that they are isolated. The way to proceed is similar to the proof that maximal sets of MUB or SIC-POVM are isolated (see Proposition \ref{prop2}). However, there is a minor additional remark: the sets are isolated despite the restricted defect $\Delta$ of the sets is non-zero. This is so because the apparently remaining $\Delta$ free parameters can be absorbed by considering a sequence of non-trivial emphasing in the Gram matrices, which means that the free parameters can be absorbed as global phases of the KS vectors. Table \ref{T3} resumes the details of our calculations.

\begin{table}[!ht]
\begin{tabular}{c|c|c|c|c|c|c}
\,$N$\,&\,$d$\,&$z$&$\tau$&$r$&$\,\Delta\,$&\# free parameters\\
\hline
13&3&24&78&66&12&0\\
18&4&63&90&83&7&0\\
21&6&105&105&103&2&0
\end{tabular}
\caption{Isolated KS sets composed by 13 vectors in dimension 3, 18 vectors in dimension 4 and 21 vectors in dimension 6. The number of zeros ($z$) appearing into the upper triangular part of the Gram matrix of KS set, total number of parameters ($\tau$), rank of the linear system defined in Eq. \ref{P2} ($r$), and restricted defect ($\Delta=\tau-r$) are defined in Section \ref{S3}. For these three KS sets the free parameters produced by a positive restricted defect $\Delta$ can be absorbed as global phases of the vectors.}
\label{T3}
\end{table}

\section{Conclusion}\label{S7}
We studied the problem to introduce free parameters in a given POVM having prescribed symmetry, where mutually unbiased bases (MUB) and symmetric informationally complete (SIC)-POVM are relevant examples (see Section \ref{S2}). In particular, our method allows us to determine whether a given quantum $t$-design having prescribed symmetry \cite{Z99} forms an isolated structure. We introduced a powerful method that divides this full non-linear problem into a linear problem and a simpler non-linear problem (see Section \ref{S3}).

Using our method, we have proven that known maximal sets of MUB in dimension 4, 8, 9 and 16 and known SIC-POVM in dimensions $4-16$ are isolated. In particular, a special class of SIC-POVM existing for 3-qubit systems, called \emph{Hoggar lines}, is isolated (see Section \ref{S4}). Moreover, we calculated an upper bound for the maximal number of free parameters that can be introduced in subsets of $2\leq m\leq d+1$ MUB in dimensions $d=2-9$ (see Section \ref{S5}). The same study has been done for equiangular tight frames in low dimensions, which define equiangular POVM (see Table \ref{T2}). 

As a further result, we studied the robustness of our method under the presence of inaccuracies in defining the generalised measurement, which allowed us to establish a confidence region for the maximal possible number of free parameters that can be introduced (see Section \ref{S4}). The importance of robustness relies on the fact that some geometrical structures, like SIC-POVM, are established analytically in low dimension only, whereas accurate numerical solutions exist in every dimension $d\leq121$ and also in $d=124, 143, 147, 168, 172, 195, 199, 228, 259$ and $323$ \cite{SG10,S17}. A Matlab source code to support calculation of the restricted defect and additional features is available on the {\rm GitHub} platform: \url{https://github.com/matrix-toolbox/defect}.

Additionally, have proven that three Kochen-Specker contextuality sets are isolated (see Section \ref{S6}). Namely, $13$ vectors in dimension $3$ \cite{SO12}, $18$ vectors in dimension $4$ \cite{CEG96} and $21$ vectors in dimension $6$ \cite{LBPC14}.

Finally, we pose some intriguing open questions: (\emph{i}) Are maximal sets of MUB isolated in every prime power dimension? (\emph{ii}) Are SIC-POVM isolated in every dimension $d>3$? Furthermore, it would be welcome to develop a more efficient software to solve the linear problem $\mathcal{P}^{(1)}_{L}$ for POVM having $N>300$ elements, e.g. maximal sets of MUB or SIC-POVM in dimension $d>16$.

\section*{Acknowledgements}
We are grateful to Markus Grassl for his advice to study the robustness of the hermitian defect under the presence of inaccuracies, to Wojciech Tadej for discussions and fruitful correspondence concerning the restricted defect of a unitary matrix and to Ingemar Bengtsson for discussions about the restricted defect for SIC-POVM in dimension 3.
Financial support by Narodowe Centrum Nauki
under the grant number DEC-2015/18/A/ST2/00274
and by the John Templeton Foundation under the project No. 56033
is gratefully acknowledged.

\appendix

\section{Robustness of restricted defect}\label{Ap1} 
In this Appendix we derive the function $f(d,N)$ which appears in Eq.(\ref{bound}) and allows us to show that the restricted defect of a unitary matrix associated to a given generalised measurement is stable with respect to small perturbations. Consider a set of vectors $\phi_j$ and the approximate vectors $\phi^{\prime}_j=\phi_j+\delta\phi_j$. The perturbed Gram matrix is given by
\begin{eqnarray}
(G+\delta G)_{ij}&=&\langle\phi_i+\delta\phi_i|\phi_j+\delta\phi_j\rangle\nonumber\\
&\approx&\langle\phi_i|\phi_j\rangle+\langle\delta\phi_i|\phi_j\rangle+\langle\phi_i|\delta\phi_j\rangle\nonumber,
\end{eqnarray}
which implies that
\begin{equation}
|\delta G_{ij}|\leq\|\phi_i\|\|\delta\phi_i\|+\|\phi_j\|\|\delta\phi_j\|\leq2\sqrt{d}\,s.
\end{equation}
Here, we used Eq.(\ref{inac}). Also, from $U=\mathbb{I}-\frac{2d}{N}G$ we have
$|\delta U_{ii}|=0$ and $|\delta U_{ij}|\leq4d^{3/2}\,s/N$ for $i\neq j$. Let us now calculate the perturbations on entries of the matrix $\mathcal{R}$, which defines the system of equations (\ref{soe}). It is simple to show that if $N>2d$ the maximal perturbations are produced by the entries of $\mathcal{R}_{jk}=-2U_{kk}U_{kj}$, associated to the left term of Eq.(\ref{soe}). Therefore
\begin{eqnarray}
|\delta \mathcal{R}_{ij}|&=&2|U_{kk}||\delta U_{ij}|\leq2\left(1-\frac{2d}{N}\right)\frac{4d^{3/2}}{N}\,s\nonumber\\
&\leq&\frac{8\,d^{3/2}}{N}\left(1-\frac{2d}{N}\right)\,s.
\end{eqnarray}
Using this result, we have
\begin{eqnarray}\label{R1}
|\delta (\mathcal{R}^{\dag}\mathcal{R})_{ij}|&=&|\delta (\mathcal{R}^{\dag})_{ij}\mathcal{R}_{ij}+(\mathcal{R}^{\dag})_{ij}\delta (\mathcal{R})_{ij}|\nonumber\\
&\leq&2\max_{\mathcal{R}_{ij}}|\delta (\mathcal{R}^{\dag})_{ij}\mathcal{R}_{ij}|\nonumber\\
&\leq&2\max_{\mathcal{R}_{ij}}|\delta (\mathcal{R}^{\dag})_{ij}|\max_{\mathcal{R}_{ij}}|\mathcal{R}_{ij}|\nonumber\\
&\leq&\frac{2^6d^{5/2}}{N^2}\left(1-\frac{2d}{N}\right)^2\sqrt{\frac{N-d}{d(N-1)}}\,s.\nonumber
\end{eqnarray}
Now we are in position to estimate the maximal perturbation on the eigenvalues of $\mathcal{R}^{\dag}\mathcal{R}$
\begin{equation}\label{R2}
\lambda^{\prime}=\lambda_k+\delta\lambda_k\approx\lambda_k+\langle\delta (\mathcal{R}^{\dag}\mathcal{R})\rangle.
\end{equation}
From the \emph{Gerschgorin circle theorem} \cite{G31} we have
\begin{equation}\label{R3}
|\langle\delta (\mathcal{R}^{\dag}\mathcal{R})\rangle|\leq\sum_{ij}|\delta (\mathcal{R}^{\dag}\mathcal{R})_{ij}|.
\end{equation}
From combining Eqs.(\ref{R1}), (\ref{R2}) and (\ref{R3}) we find that
$|\lambda^{\prime}_i-\lambda_i|\leq f(d,N)\,s$, where
\begin{equation*}
f(d,N)=\frac{2^6d^{5/2}}{N^2}\left(1-\frac{2d}{N}\right)^2\sqrt{\frac{N-d}{d(N-1)}}\,s.
\end{equation*}
Given that $\mathcal{R}^{\dag}\mathcal{R}$ is a positive operator, its eigenvalues $\lambda_i$ coincide with its singular values $\sigma_i$. Therefore $|\sigma^{\prime}_i-\sigma_i|\leq f(d,N)\,s$, which proves Eq.(\ref{bound}).

\end{document}